\documentclass[fleqn,11pt,twoside]{article}

\usepackage{amsthm,amsthm,amssymb, color, xcolor,epsfig, graphics, subfigure}

\usepackage{amsmath, graphicx, latexsym, lscape}
\usepackage{amscd}

\makeatletter
\newcommand{\copyrightnote}[2]{{\renewcommand{\thefootnote}{}
 \footnotetext{\small\it
\begin{flushleft}
 \copyright \ #1   #2  
\end{flushleft}}}}

\newcommand{\Name}[1]{\begin{flushleft}
                       \LARGE \bf #1
                       \end{flushleft}\vspace{-3mm}}

\newcommand{\Author}[1]{\begin{flushleft}
                       \it #1 \end{flushleft}}

\newcommand{\Address}[1]{\begin{flushleft}
                       \it #1 \end{flushleft}}

\newcommand{\Date}[1]{\begin{flushleft}
                      \small  \it #1 \end{flushleft}}

\newcommand{\evenhead}{Author \ name}
\newcommand{\oddhead}{Article \ name}

\renewcommand{\@evenhead}{
\hspace*{-3pt}\raisebox{-15pt}[\headheight][0pt]{\vbox{\hbox to \textwidth
{\thepage \hfil \evenhead}\vskip4pt \hrule}}}
\renewcommand{\@oddhead}{
\hspace*{-3pt}\raisebox{-15pt}[\headheight][0pt]{\vbox{\hbox to \textwidth
{\oddhead \hfil \thepage}\vskip4pt\hrule}}}
\renewcommand{\@evenfoot}{}
\renewcommand{\@oddfoot}{}

\long\def\@makecaption#1#2{%
  \vskip\abovecaptionskip
  \sbox\@tempboxa{\small \textbf{#1.}\ \ #2}%
  \ifdim \wd\@tempboxa >\hsize
    {\small \textbf{#1.}\ \ #2}\par
  \else
    \global \@minipagefalse
    \hb@xt@\hsize{\hfil\box\@tempboxa\hfil}%
  \fi
  \vskip\belowcaptionskip}

\newcommand{\JNMPnumberwithin}[3][\arabic]{%
  \@ifundefined{c@#2}{\@nocounterr{#2}}{%
    \@ifundefined{c@#3}{\@nocnterr{#3}}{%
      \@addtoreset{#2}{#3}%
      \@xp\xdef\csname the#2\endcsname{%
        \@xp\@nx\csname the#3\endcsname .\@nx#1{#2}}}}%
}

\renewenvironment{proof}[1][\proofname]{\par
  \normalfont
  \topsep6\p@\@plus6\p@ \trivlist
  \item[\hskip\labelsep\textbf{%
    #1\@addpunct{.}}]\ignorespaces
}{%
  \qed\endtrivlist
}

\newcommand{\resetfootnoterule} {
  \renewcommand\footnoterule{%
  \kern-3\p@
  \hrule\@width.4\columnwidth
  \kern2.6\p@}
}

\renewcommand{\footnoterule}{}

\makeatother

\theoremstyle{definition}

\theoremstyle{plain}
\newtheorem{Th}{Theorem}[section]
\newtheorem{Cor}[Th]{Corollary}
\newtheorem{Lem}[Th]{Lemma}
\newtheorem{Prop}[Th]{Proposition}
\theoremstyle{definition}
\newtheorem{Def}{Definition}[section]

\theoremstyle{remark}
\newtheorem*{Rem}{Remark}
\numberwithin{equation}{section}

\newcommand{\EE}{{\mathbb E}}

\newcommand{\ZZ}{{\mathbb Z}}

\newcommand{\RR}{{\mathbb R}}

\newcommand{\be}{\boldsymbol{e}}
\newcommand{\bn}{\boldsymbol{n}}
\newcommand{\br}{\boldsymbol{r}}

\newcommand{\bu}{\boldsymbol{u}}
\newcommand{\bv}{\boldsymbol{v}}

\newcommand{\bY}{{\boldsymbol Y}}

\newcommand{\bX}{\boldsymbol{X}}

\newcommand{\bTh}{\boldsymbol{\Theta}}

\newcommand{\bnu}{\boldsymbol{\nu}}

\newcommand{\bpsi}{\boldsymbol{\psi}}

\begin{document}

\renewcommand{\evenhead}{ {\LARGE\textcolor{blue!10!black!40!green}{{\sf \ \ \ ]ocnmp[}}}\strut\hfill A Doliwa}
\renewcommand{\oddhead}{ {\LARGE\textcolor{blue!10!black!40!green}{{\sf ]ocnmp[}}}\ \ \ \ \  Bäcklund transformations as integrable discretization}

\thispagestyle{empty}
\newcommand{\FistPageHead}[3]{
\begin{flushleft}
\raisebox{8mm}[0pt][0pt]
{\footnotesize \sf
\parbox{150mm}{{Open Communications in Nonlinear Mathematical Physics}\ \ \ \ {\LARGE\textcolor{blue!10!black!40!green}{]ocnmp[}}
\quad Special Issue 1, 2024\ \  pp
#2\hfill {\sc #3}}}\vspace{-13mm}
\end{flushleft}}

\FistPageHead{1}{\pageref{firstpage}--\pageref{lastpage}}{ \ \ }

\strut\hfill

\strut\hfill

\copyrightnote{The author(s). Distributed under a Creative Commons Attribution 4.0 International License}

\begin{center}
{  {\bf This article is part of an OCNMP Special Issue\\ 
\smallskip
in Memory of Professor Decio Levi}}
\end{center}

\smallskip

\Name{B\"{a}cklund transformations as integrable discretization. The geometric approach}

\Author{Adam Doliwa}

\Address{Faculty of Mathematics and Computer Science\\
	University of Warmia and Mazury in Olsztyn\\
	ul.~S{\l}oneczna~54, 10-710~Olsztyn, Poland}

\Date{Received August 31, 2023; Accepted November 27, 2023}

\smallskip

\textit{To memory of Decio Levi and Antoni Sym, distinguished researchers and friends}

\setcounter{equation}{0}

\begin{abstract}

\noindent 
We present interpretation of known results in the theory of discrete asymptotic and discrete conjugate nets from the \emph{discretization by B\"{a}cklund transformations} point of view. We collect both classical formulas of XIXth century differential geometry of surfaces and their transformations, and more recent results from geometric theory of integrable discrete equations. We first present  transformations of hyperbolic surfaces within the context of the Moutard equation and  Weingarten congruences. The permutability property of the transformations provides a way to construct integrable discrete analogs of the asymptotic nets for such surfaces. Then after presenting the theory of conjugate nets and their transformations we apply the principle that  B\"{a}cklund transformations provide integrable discretization to obtain known results on the discrete conjugate nets. The same approach gives, via the Ribaucour transformations, discrete integrable analogs of orthogonal conjugate nets. 

\end{abstract}

\label{firstpage}


\section{Introduction}
Given integrable system of differential equations one is often interested in finding the corresponding (in the sense of small lattice step size limit) discrete system while preserving the integrability properties. It turns out that usually the simple/naive replacement of derivatives by difference operators spoils the integrability. The discretization has to be made on the level where the integrability features are visible and transparent. Such methods are, for example (i) discrete version of the linear problem~\cite{Abl-Lad}, (ii) the Hirota method via a bilinear form~\cite{Hirota-KdV}, (iii)~extensions of the Zakharov--Shabat dressing method~\cite{LeviPilloniSantini}, (iv)~direct linearization using linear integral equations~\cite{NijhoffQuispelCapel}.

Another technique, which is the subject of the present work,  is based on B\"{a}cklund transformations, which are discrete symmetries of integrable equations. It is the fundamental observation made by Decio Levi, which we present quoting abstracts of two of his papers:

\emph{It is shown that any B\"{a}cklund transformation
	of a nonlinear differential equation integrable by the multichannel Schr\"{o}dinger eigenvalue problem can be written in the form $V_{x}= U^\prime V - VU$. This allows us to interpret the B\"{a}cklund
	transformation formally as a nonlinear differential difference
	equation for which we can immediately construct the soliton
	solutions.} \cite{LeBen}

\emph{In this paper, one shows that the best known nonlinear differential difference
	equations associated with the discrete Schr\"{o}dinger spectral problem and also with the
	discrete Zakharov-Shabat spectral problem can be interpreted as B\"{a}cklund transformations for some continuous nonlinear evolution equations.} \cite{Levi}

B\"{a}cklund transformations arose in connection with the construction, by XIXth century geometers~\cite{Bianchi-hab,Backlund}, of pseudospherical surfaces and corresponding solutions of the sine-Gordon equation. It was shown by Bianchi~\cite{Bianchi-permut} that such transformations can be iterated leading to an algebraic superposition formula.
More recently, Wahlquist and Estabrook~\cite{WahlquistEstabrook} demonstrated that also the Korteweg--de~Vries equation, which is a paradigmatic example of integrable partial differential equation~\cite{GGKM}, admits invariance under a B\"{a}cklund-type transformation and possesses an associated
permutability theorem. They have used for that purpose the transformation introduced by Darboux~\cite{Darboux-transf} in the context of Sturm--Liouville problems. The subject is generally known in the soliton theory as B\"{a}cklund or Darboux transformations~\cite{Lamb,RogersShadwick,ms,RogersSchief, GuHuZhou,Ferapontov}, but in the geometric theory of transformations of surfaces exhibiting permutability property, also other names are relevant~\cite{Demoulin-R, Ribaucour, Jonas}. The classical results  of the old differential geometry of surfaces and their transformations are summarized in \cite{DarbouxIV,Darboux-OS,Bianchi,Eisenhart-TS,Tzitzeica,Lane,Finikov}.

The most general transformations of conjugate nets and their permutability were introduced and studied by Jonas~\cite{Jonas,Eisenhart-TS}. The Darboux equations of multidimensional conjugate nets~\cite{Darboux-OS} have been rediscovered by  Zakharov and Manakov~\cite{ZaMa1,ZaMa2} as the most general systems solvable by the non-local $\bar\partial$-dressing method. The discrete analogue of conjugate nets on a surface was introduced first on the geometric level~\cite{Sauer2,Sauer}, and connected to integrability theory in~\cite{DCN}. The integrable discrete version of the corresponding Darboux system was given in~\cite{BoKo}. Geometric studies of multidimensional discrete conjugate nets have been initiated in~\cite{MQL} and were followed in~\cite{MDS,DMS,TQL}. The integrability of circular lattices, which form a distinguished reduction of discrete conjugate nets and in the continuous limit give orthogonal conjugate nets, was first studied geometrically in \cite{CDS} and then confirmed by other tools \cite{DMS,KoSchief2,AKV,DMM}. Compelling reasons for such an interpretation were given in \cite{2dcl1,2dcl2} on the basis of the computer graphics, and in \cite{BP2,Bobenko-O} from the theory of discrete isothermic nets.

I met first time Decio Levi in late eighties when he came to Institute of Theoretical Physics of Warsaw University to visit his friend Antoni Sym, who supervised both my master and PhD theses. They worked together on integrable generalization of pseudospherical surfaces, known now as Bianchi surfaces~\cite{LeviSym-B}. In the present work we show a harmonious coexistence of two points of view: \emph{B\"{a}cklund transformations provide integrable discretization} by Decio Levi, and \emph{soliton theory is surface theory} \cite{Sym} by Antoni Sym. In the theory of integrable systems it is quite common  that the same equations can be derived using different methods, which give different perspective and emphasize different connections.

The paper is constructed as follows. In Section~\ref{sec:dis-ass} we present the classical theory of hyperbolic surfaces in asymptotic parametrization emphasizing their connection with the Moutard equation~\cite{Moutard}. The corresponding transformations and their permutability property~\cite{NiSchief} give rise to discrete asymptotic nets, which coincide with their natural geometric analogs~\cite{Sauer}. The next Section~\ref{sec:conj}  is devoted to conjugate nets and their transformations. We give a discretization of the nets starting from their fundamental transformation and exploiting its permutability properties. In Section~\ref{sec:circ} we present derivation of integrable discrete version of orthogonal conjugate nets on the base of geometric interpretation of the Ribaucour reduction~\cite{Ribaucour,Demoulin-R} of the fundamental  transformation. We conclude the paper with additional discussion on geometry of Hirota's discrete Kadomtsev--Petviashvili (KP) system and with some remarks about dispersionless systems.

In the paper we tried to present old results, most of them more than one hundred years old, from contemporary perspective and in unified notation. We remark that  interpretation of the fundamental transformation and its Ribaucour reduction in terms of vertex operators within the free-fermion formalism of the multicomponent KP hierarchy was the subject of \cite{DMMMS,DMM}, where also the aspects of integrable discretization of conjugate and orthogonal nets were investigated.

\section{Discretization of asymptotic nets, and the Moutard \\ transformation} \label{sec:dis-ass}
The classical transformation of Bianchi and B\"{a}cklund for pseudospherical surfaces and the sine-Gordon equation can be considered as a reduction of Weingarten transformation of hyperbolic surfaces in asymptotic parametrization. We devote the present Section to such asymptotic nets and to the Moutard equation~\cite{Moutard}, which governs the behaviour of their normal vector. We use the standard notation and terminology of the classical theory of surfaces~\cite{Eisenhart-TCS}, see also~\cite{RogersSchief}.

\subsection{Hyperbolic surfaces, the Moutard equation, and the Lelieuvre\\ formulas}
Let $(u,v)$ be local coordinate system on a surface $\Sigma$ in $\RR^3$, and by $\br(u,v)$ denote the position vector of a generic point. The coordinate lines are called asymptotic when in every point their tangent planes coincide with the tangent plane to the surface. Surfaces admitting the asymptotic coordinates are called hyperbolic. In such case we have
\begin{align} \label{eq:r-uu}
\br_{,uu} & = a_1 \br_{,u} + b_1 \br_{,v}\; , \\
\label{eq:r-vv}
\br_{,vv} & = a_2 \br_{,u} + b_2 \br_{,v}\; ,
\end{align}
where $a_i$, and $b_i$, $i=1,2$, are functions of the local coordinates; here and in all the paper by a subscript after comma we denote the partial derivative with respect to the corresponding variable.
As a consequence of the compatibility condition $\br_{,uuvv} = \br_{,vvuu}$ there exists a function $\phi$, given up to an additive constant, such that
\begin{equation} \label{eq:ab-phi}
a_1 = \phi_{,u}\;, \qquad b_2 = \phi_{,v}\; .
\end{equation}
By direct calculation one can check that
the normal vector 
\begin{equation} \label{eq:bnu}
\bnu = \mathrm{e}^{-\phi} \br_{,u} \times \br_{,v}
\end{equation}
satisfies the Moutard equation
\begin{equation} \label{eq:M-uv}
\bnu_{,uv} = f \bnu \; , 
\end{equation}
with the potential
\begin{equation}
f = \phi_{,uv} - b_1 a_2 \; .
\end{equation}
Moreover, using eventually the allowed freedom in definition of $\phi$ and/or changing the orientation, the position vector is given by the Lelieuvre formulas~\cite{Lelieuvre}
\begin{equation} \label{eq:L-r-u-v}
\br_{,u}  = \bnu_{,u} \times \bnu \; , \qquad
\br_{,v}  = \bnu \times \bnu_{,v} \; .
\end{equation}

\subsection{The Moutard transformation and its permutability property}
The following result provides a way how to transform solutions of a Moutard equation into new solution of equation of the same form but with different potential. In consequence, given a hyperbolic surface it allows to construct new surface of the same type.
\begin{Th}{\cite{Moutard}}
	Given vector-valued solution $\bnu$ of the Moutard equation corresponding to a given potential $f$, and given scalar solution $\theta$ of the same equation, then the vector-valued function $\hat\bnu$ defined by the compatible equations
	\begin{align} \label{eq:MT-u}
	(\theta\hat\bnu)_{,u} & = \;\; \theta_{,u} \bnu - \theta \bnu_{,u} \; ,\\ \label{eq:MT-v}
	(\theta\hat\bnu)_{,v} & = - \theta_{,v} \bnu + \theta \bnu_{,v} \; ,
	\end{align}
	satisfies the Moutard equation with the potential
	\begin{equation*}
	\hat{f} = \frac{\hat\theta_{,uv}}{\hat\theta} \; , \qquad \text{where} \qquad  \hat\theta = \frac{1}{\theta} \; .
	\end{equation*}
\end{Th}
\begin{Cor} \label{cor:M-1}
	Notice that $\bnu$ is the Moutard transform of $\hat\bnu$ with $\hat\theta$ taken as the transformation function.
\end{Cor}

\begin{Cor} \label{cor:hat-r}
	One can check that the surface 
	\begin{equation} \label{eq:hat-r}
	\hat\br = \br + \hat\bnu \times \bnu,
	\end{equation}
	can be obtained from $\hat\bnu$ via the Lelieuvre formulas~\eqref{eq:L-r-u-v}. In particular the local parameters $(u,v)$ form an asymptotic coordinate system on the transformed surface.
\end{Cor}

\begin{Rem}
	The two-parameter family of lines $\langle \br, \hat\br \rangle$ , which are tangent to both surfaces $\Sigma$ and $\hat\Sigma$ at the corresponding points, forms the so called Weingarten congruence. The corresponding transformation between hyperbolic surfaces $\Sigma$ and $\hat\Sigma$ is called the Weingarten transformation.
\end{Rem}

Let $\theta^{i}$, $i=1,2$, be two solutions of the Moutard equation satisfied by the normal $\bnu$, and let $\bnu^{\{i\}}$, $i=1,2$, be its corresponding two transforms. By $\theta^{2\{1\}}$ denote also the Moutard transform of $\theta^2$ with respect to $\theta^1$ 
\begin{equation} \label{eq:MS-1}
\begin{pmatrix}  \theta^1 \begin{pmatrix} \bnu^{\{1\}} \\ \theta^{2\{1\}} \end{pmatrix}\end{pmatrix}_{,u} = 
\theta^1_{,u} \begin{pmatrix} \bnu \\ \theta^{2} \end{pmatrix} - \theta^1 \begin{pmatrix} \bnu \\ \theta^{2} \end{pmatrix}_{,u},\qquad
\begin{pmatrix}  \theta^1 \begin{pmatrix} \bnu^{\{1\}} \\ \theta^{2\{1\}} \end{pmatrix}\end{pmatrix}_{,v} =
- \theta^1_{,v} \begin{pmatrix} \bnu \\ \theta^{2} \end{pmatrix} + \theta^1 \begin{pmatrix} \bnu \\ \theta^{2} \end{pmatrix}_{,v},
\end{equation}
and by $\theta^{1\{2\}}$ denote the Moutard transform of $\theta^1$ with respect to $\theta^2$
\begin{equation} \label{eq:MS-2}
\begin{pmatrix}  \theta^2 \begin{pmatrix} \bnu^{\{2\}} \\ \theta^{1\{2\}} \end{pmatrix}\end{pmatrix}_{,u} =
\theta^2_{,u} \begin{pmatrix} \bnu \\ \theta^{1} \end{pmatrix} - \theta^2 \begin{pmatrix} \bnu \\ \theta^{1} \end{pmatrix}_{,u},\qquad
\begin{pmatrix}  \theta^2 \begin{pmatrix} \bnu^{\{2\}} \\ \theta^{1\{2\}} \end{pmatrix}\end{pmatrix}_{,v} = -
\theta^2_{,v} \begin{pmatrix} \bnu \\ \theta^{1} \end{pmatrix} + \theta^2 \begin{pmatrix} \bnu \\ \theta^{1} \end{pmatrix}_{,v}.
\end{equation}
In consequence, the transformation formulas \eqref{eq:hat-r} on the hyperbolic surfaces level $\Sigma^{\{1\}}$ and $\Sigma^{\{2\}}$ take the form
\begin{equation} \label{eq:hat-r-1-2}
\br^{\{1\}}  = \br + \bnu^{\{1\}} \times \bnu,\qquad
\br^{\{2\}}  = \br  + \bnu^{\{2\}} \times \bnu.
\end{equation}

Let us apply to $\bnu^{\{1\}}$ the Moutard transformation with $\theta^{2\{1\}}$,
\begin{equation} \label{eq:MS-12}
\left(  \theta^{2\{1\}} \bnu^{\{1, 2\}}\right)_{,u} =
\theta^{2\{1\}}_{,u} \bnu^{\{1\}} - \theta^{2\{1\}}  \bnu^{\{1\}}_{,u},\qquad
\left(  \theta^{2\{1\}} \bnu^{\{1, 2\}}\right)_{,v} =
- \theta^{2\{1\}}_{,v} \bnu^{\{1\}} + \theta^{2\{1\}}  \bnu^{\{1\}}_{,v},
\end{equation}
and to $\bnu^{\{2\}}$ the transformation with $\theta^{1\{2\}}$
\begin{equation} \label{eq:MS-21}
\left(  \theta^{1\{2\}} \bnu^{\{2,1 \}}\right)_{,u} =
\theta^{1\{2\}}_{,u} \bnu^{\{2\}} - \theta^{1\{2\}}  \bnu^{\{2\}}_{,u},\qquad
\left(  \theta^{1\{2\}} \bnu^{\{2, 1\}}\right)_{,v} =
- \theta^{1\{2\}}_{,v} \bnu^{\{2\}} + \theta^{1\{2\}}  \bnu^{\{2\}}_{,v}.
\end{equation}
We would like both transformations give the same result $\bnu^{\{2,1 \}} = \bnu^{\{1,2 \}}$. Moreover since we have
\begin{equation*}
(\theta^1 \theta^{2\{1\}} + \theta^2 \theta^{1\{2\}})_{,u} =0 =
(\theta^1 \theta^{2\{1\}} + \theta^2 \theta^{1\{2\}})_{,v} ,
\end{equation*}
then the additive constants in definition of $\theta^1 \theta^{2\{1\}}$ and of $\theta^2 \theta^{1\{2\}}$ can be fixed such that
\begin{equation} \label{eq:MS-theta}
\theta^2 \theta^{1\{2\}}= \theta^{12} = -\theta^1 \theta^{2\{1\}} \; .
\end{equation}
By elimination of derivatives of the normal vectors from equations \eqref{eq:MS-1}-\eqref{eq:MS-2} and  \eqref{eq:MS-12}-\eqref{eq:MS-21}, and using identity \eqref{eq:MS-theta}  we can arrive to the the following important result. 
\begin{Th}[Permutability of the Moutard transformations]
	The vector-valued functions $\bnu^{\{12\}}$ given by algebraic formula
	\begin{equation} \label{eq:M-12}
	\bnu^{\{12\}} - \bnu = \frac{\theta^1 \theta^2}{\theta^{12}}\left( \bnu^{\{1\}} - \bnu^{\{2\}}\right)
	\end{equation}
	is simultaneously Moutard transform of $\bnu^{\{1\}}$ with respect to $\theta^{2\{1\}}$ and the Moutard transform of $\bnu^{\{2\}}$ with respect to $\theta^{1\{2\}}$.  
	\begin{equation*}
	\begin{CD}
	\bnu @>{\theta^1}>> \bnu^{\{1\}} \\
	@V{\theta^2}VV    @VV{\theta^{2\{1\}}}V  \\
	\bnu^{\{2\}} @>{\theta^{1\{2\}}}>>  \bnu^{\{12\}}.
	\end{CD}
	\end{equation*} 
\end{Th}
\begin{Rem}
	Actually, because of free additive parameter in definition of $\theta^{12}$ we obtain this way one-parameter family of the transforms.
\end{Rem}
\begin{Cor}
	The one-parameter family of vector-valued functions $\br^{\{12\}}$ given by algebraic formula
	\begin{equation} \label{eq:r-12}
	\br^{\{12\}} = \br + \frac{\theta^1 \theta^2}{\theta^{12}} \bnu^{\{1\}}\times \bnu^{\{2\}}
	\end{equation}
	provides simultaneously the Weingarten transforms of the hyperbolic surface $\br^{\{1\}}$ with respect to $\theta^{2\{1\}}$ and the Weingarten transforms of the hyperbolic surface $\br^{\{2\}}$ with respect to $\theta^{1\{2\}}$. 
\end{Cor}

\subsection{The Moutard transformation as integrable discretization}
The successive application of the Weingarten-Moutard transforms, taking into account their algebraic superposition formula, to a given hyperbolic surface $\Sigma$ allows to build a two dimensional lattice $\Sigma^{m,n}$ of such surfaces. Let us fix a point on $\Sigma$ and trace properties of a lattice in $\RR^3$ of the corresponding points, represented by vectors $\br^{m,n}$. 
\begin{Prop} \label{prop:dis-as}
	The five points $\br^{m,n}$, $\br^{m\pm1,n}$, $\br^{m,n\pm 1}$ belong to a common plane.
\end{Prop}
\begin{proof}
	The statement follows from equations \eqref{eq:hat-r-1-2}, which imply that lines $\langle \br^{m,n}, \br^{m+1,n} \rangle$ are orthogonal to both the vectors $\bnu^{m,n}$ and $\bnu^{m+1,n}$, and the lines $\langle \br^{m,n}, \br^{m,n+1} \rangle$ are orthogonal to both the vectors $\bnu^{m,n}$ and $\bnu^{m,n+1}$.
	The common plane of the five points is orthogonal to the vector $\bnu^{m,n}$.
\end{proof}
\begin{Cor}
	Equations \eqref{eq:hat-r-1-2} are discrete analogs of the Lelieuvre formulas \eqref{eq:L-r-u-v}, see also~\cite{Kon-Pin}.
\end{Cor}
\begin{Rem}
	The three points $\br^{m,n}$, $\br^{m\pm1,n}$, define the tangent plane of the first discrete coordinate curve at point $\br^{m,n}$, and the three points $\br^{m,n}$, $\br^{m,n\pm 1}$, define the tangent plane of the second curve at this point. Both planes coincide with the common plane of Proposition~\ref{prop:dis-as}, what allows to call the curves the discrete asymptotic lines. 
\end{Rem}

\begin{Def}[\cite{Sauer,Sauer2}]
	The discrete asymptotic net is a map $\br \colon \ZZ^2 \to \RR^3$ such that for arbitrary $(m,n) \in \ZZ^2$ the five points $\br^{m,n}$, $\br^{m\pm1,n}$, $\br^{m,n\pm 1}$ are coplanar. 
\end{Def}
\begin{figure}
	\begin{center}
		\includegraphics[width=3.5cm]{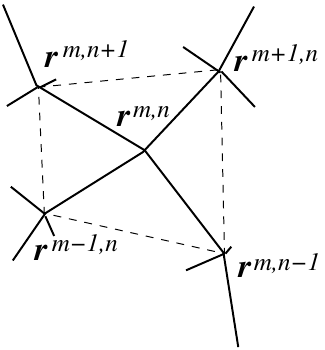}
	\end{center}
	\caption{Discrete asymptotic nets}
	\label{fig:d-as}
\end{figure} 
\begin{Cor}
	The algebraic superposition formula \eqref{eq:M-12} of the Moutard transformations on the level of the normal vector $\bnu^{m,n}$ can be interpreted as the discrete Moutard equation~\cite{NiSchief}
	\begin{equation}
	\label{eq:M-mn}
	\bnu^{m+1,n+1} - \bnu^{m,n} = f^{m,n} \left( \bnu^{m+1,n} - \bnu^{m,n+1} \right).
	\end{equation}
\end{Cor}
\begin{Rem}
	To obtain the Moutard equation \eqref{eq:M-uv} and the Lelieuvre formulas~\eqref{eq:L-r-u-v} from \eqref{eq:M-mn} and \eqref{eq:hat-r-1-2} one first has to change the orientation of the normal vector according to $\bnu^{m,n} \to (-1)^n \bnu^{m,n}$.
\end{Rem}

Integrability of the hyperbolic surfaces can be understood as existence of transformations satisfying algebraic permutability principle. It turns out that one can construct analogous transformations for integrable discrete hyperbolic surfaces (discrete asymptotic nets). The corresponding transformation formulas, their geometric interpretation, and their permutability property have been considered in \cite{NiSchief,Nieszporski-dan,AD-W}. See also \cite{HVR} for development of their theory in direction of application in computer graphic. 

Hyperbolic surfaces with constant Gauss curvature, i.e. pseudospherical surfaces, are described by the celebrated sine-Gordon equation for the angle between the asymptotic coordinates on such surfaces. Restricting the Weingarten--Moutard transformation to such surfaces we obtain discrete asymptotic nets whose elementary quadrilaterals have opposite sides of equal length, and which provide integrable discrete analogs of pseudospherical surfaces~\cite{Sauer-pss,Wunderlich,Bobenko-Pinkall-pss}. The corresponding Bianchi's permutability theorem for the B\"{a}cklund transformation of the sine-Gordon equation provides its integrable difference analog \cite{Hirota-dsG}, and describes the angle between asymptotic coordinates on discrete pseudospherical surface.
The Bianchi surfaces~\cite{Bianchi}, whose integrability was discussed in~\cite{LeviSym-B}, are hyperbolic surfaces characterized by the property that their Weingarten transformation preserves Gauss curvature in the corresponding points; see~\cite{Bob-Schief,Schief-Calapso,DNS-ass,Schief-JGP114} for discussion of discrete analogs and this and other integrable reductions of asymptotic nets.

\subsection{Discrete BKP equation}
Integrable non-linear equations appear in the context of discrete Moutard equation, prior to reductions of discrete hyperbolic surfaces, for more then two discrete variables \cite{NiSchief,BQL}. Consider a map $\bnu \colon \ZZ^N \to \RR^M$, $N, M \geq 3$, which satisfies the system of discrete Moutard equations in each pair of variables
\begin{equation} \label{eq:Mout-N}
\bnu^{\bn + \be_i + \be_j} - \bnu^{\bn} = f_{ij}^{\bn} \left( \bnu^{\bn+\be_i} - \bnu^{\bn+\be_j} \right), \qquad 1\leq i \neq j \leq N;
\end{equation}
here $\bn\in\ZZ^N$, and $\be_i$ is the unit vector in $i$th direction, $i=1,\dots ,N$.
Compatibility of the system leads to the following set of nonlinear equations
\begin{equation} \label{eq:nonl-BQL-f}
1 + f_{jk}^{\bn+\be_i}(f_{ij}^{\bn} - f_{ik}^{\bn}) = f_{ik}^{\bn + \be_j} f_{ij}^{\bn} = f_{ij}^{\bn+\be_k}f_{ik}^{\bn}, 
\qquad i,j,k \quad \text{distinct}, \qquad f_{ji}^{\bn}=-f_{ij}^{\bn}.
\end{equation}
The second equality  
implies existence of the potential $\tau:\ZZ^N\to\RR$, 
in terms of which the functions $f_{ij}$ can be written as
\begin{equation} \label{eq:tau}
f_{ij}^{\bn} = \frac{\tau^{\bn+\be_i}\tau^{\bn + \be_j}}{\tau^{\bn} \, \tau^{\bn + \be_i + \be_j}}, \qquad i < j .
\end{equation}
The first equality can be then rewritten in the form of the 
system of Miwa's discrete 
BKP equations \cite{Miwa}
\begin{equation} \label{eq:BKP-nlin}
\tau^{\bn} \tau^{\bn+\be_i + \be_j + \be_k} = \tau^{\bn + \be_i + \be_j}\tau^{\bn + \be_k} - \tau^{\bn+ \be_i + \be_k}\tau^{\bn+\be_j} + 
\tau^{\bn+\be_j + \be_k}\tau^{\bn+\be_i}, \quad 1\leq i< j < k \leq N.
\end{equation}
Geometric meaning of the above equations goes beyond the theory of discrete asymptotic nets, and can be incorporated~\cite{BQL} into the theory of discrete conjugate nets~\cite{MQL,TQL}.

\section{Multidimensional conjugate nets, their fundamental \\ transformation, and multidimensional lattices \\of planar quadrilaterals} \label{sec:conj}  
Conjugate nets on a surface are second, after asymptotic nets, distinguished coordinate systems studied in depth by geometers of XIXth century. They include, as a special subcase, curvature coordinates, whose theory will be studied in the next Section. For classical results on the subject see works of Gabriel Lam\'{e}~\cite{Lame}, Luigi Bianchi~\cite{Bianchi} or Gaston Darboux~\cite{Darboux-OS}. 
The relation of conjugate nets on a surface to linear partial differential equations of the second order allows to transfer special theorems on such equations~\cite{Laplace,Levy,Moutard,Goursat} to the geometric level. The theory of transformations within special classes of such equations/conjugate nets has obtained mature form in works of Hans Jonas~\cite{Jonas} and Luther P. Eisenhart~\cite{Eisenhart-TS}, see also~\cite{Tzitzeica,Lane,Finikov}.
The Darboux equations, which describe multidimensional conjugate nets, were rediscovered in~\cite{ZaMa1,ZaMa2} in the context of soliton theory as the most general partial differential equations integrable by the non-local $\bar\partial$-dressing method. Moreover, in \cite{KvL} they were isolated as the simplest equations within multicomponent KP hierarchy. We start this Section with presenting basic elements of the theory of conjugate nets on a surface and their transformations. Then we move to the multidimensional nets.

\subsection{Conjugate coordinates on a surface, and the L\'{e}vy transformation}
Local coordinates $(u,v)$, on a surface in the space $\RR^N$ of arbitrary dimension $N\geq 3$, are called conjugate if the second mixed derivative of the position vector $\br(u,v)$ are tangent to the surface. The defining equation takes then the form of the \emph{Laplace equation}
\begin{equation} \label{eq:L-uv}
\br_{,uv} = a \br_{,u} + b \br_{,v},
\end{equation} 
where $a(u,v)$ and $b(u,v)$ are corresponding functions of the conjugate parameters. The following considerations allow to introduce concepts relevant in the general theory of transformations of multidimensional conjugate nets.

Let $\theta$ be a scalar solution of equation \eqref{eq:L-uv}, linearly independent of the components of $\br$, then by direct calculation one can check that the so called \emph{L\'{e}vy transforms} of $\br$, given by
\begin{equation}
\br^{(u)} = \br - \frac{\theta}{\theta_{,u}} \br_{,u}, \qquad
\br^{(v)} = \br - \frac{\theta}{\theta_{,v}} \br_{,v}, 
\end{equation}
are new surfaces with $(u,v)$ being local conjugate coordinates. The corresponding Laplace equations~\eqref{eq:L-uv} of the new nets have coefficients
\begin{align}
a^{(u)} = a + \left(\log\frac{\theta}{\theta_{,u}}\right)_{,v}, \qquad & b^{(u)} = b^{(v)} + \left(\log a^{(u)}\right)_{,u},\\
a^{(v)} = a^{(u)} + \left(\log b^{(v)}\right)_{,v}, \qquad & b^{(v)} = b + \left(\log\frac{\theta}{\theta_{,v}}\right)_{,u},
\end{align}
what can be verified by direct calculation.
\begin{figure}
	\begin{center}
		\includegraphics[width=8cm]{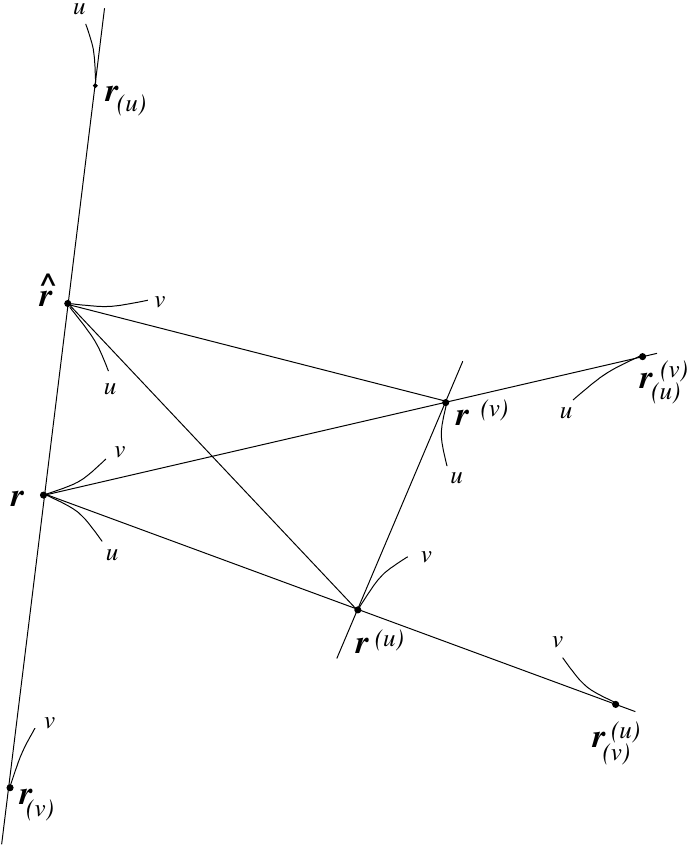}
	\end{center}
	\caption{The fundamental transform $\hat{\br}$, the L\'{e}vy transforms $\br^{(u)}$, $\br^{(v)}$, the adjoint L\'{e}vy transforms $\br_{(u)}$, $\br_{(v)} $, and the Laplace transforms $\br^{(u)}_{(v)}$, $\br^{(v)}_{(u)} $ of two-dimensional conjugate net $\br$}
	\label{fig:fund-uv}
\end{figure} 
In the proof it is convenient to show first that
\begin{equation*}
\br^{(u)}_{,v} = a^{(u)} \left( \br^{(u)} - \br^{(v)} \right), \qquad 
\br^{(v)}_{,u} = b^{(v)} \left( \br^{(v)} - \br^{(u)} \right),
\end{equation*}
what means that the lines joining the corresponding points of both L\'{e}vy transforms of the conjugate net $\br$ are simultaneously tangent to $v$-coordinate lines of $\br^{(u)}$ and to $u$-coordinate lines of $\br^{(v)}$, see Figure~\ref{fig:fund-uv}.

Few comments are in order:
\begin{itemize}
	\item Lines through $\br$ in direction of $\br_{,u}$ form the \emph{$u$-tangent congruence} of the conjugate net. Equivalently, $\br$ is \emph{$u$-focal net} of the congruence. Similarly one defines the $v$-tangent congruence of $\br$, which is its $v$-focal net. Using such a terminology one can state that $\br^{(v)}$ is $u$-focal net of the $v$-tangent congruence of $\br^{(u)}$, and \emph{vice versa}.
	\item Notice that we use the notion of congruence of lines in the narrow sense, i.e. the parameters $(u,v)$ of the family define its focal nets. This means that one-dimensional family of lines parametrized by $u$  forms a developable surface and is made of tangents to $u$-coordinate on the $u$-focal net; similar condition folds for $v$-parameter family of the congruence. We say that the congruence is referred to its developables.
	\item The unique $v$-focal net of the $u$-tangent congruence to $\br$, given by
	\begin{equation}
	\br^{(u)}_{(v)} = \br -\frac{1}{b} \br_{,u}, 
	\end{equation} 
	is called the \emph{$uv$-Laplace transform of $\br$}. Analogously one defines the $vu$-Laplace transform of $\br$
	\begin{equation}
	\br^{(v)}_{(u)} = \br -\frac{1}{a} \br_{,v}.
	\end{equation}  
	In particular the conjugate net $\br^{(u)}$ is the $uv$-Laplace transform  of $\br^{(v)}$. 
	\item The net $\br^{(u)}$ is called \emph{conjugate} to the $u$-tangent congruence of $\br$, similarly the net $\br^{(v)}$ is conjugate to the $v$-tangent congruence of $\br$. In such  relation the conjugate coordinates on a surface are the parameters which define the focal nets of the congruence. One can reverse the situation and try to find focal nets of a congruence conjugate to a given conjugate net. Such focal nets, denoted by $\br_{(u)}$ and $\br_{(v)}$, are called \emph{adjoint L\'{e}vy transforms} of the net.
	\item The conjugate net $\hat{\br}$ is called a \emph{fundamental transform} of $\br$ when the two-dimensional family of lines joining their corresponding points forms a congruence (in the narrow sense explained above) whose developables cut both nets along the conjugate coordinate lines, i.e. both nets are conjugate to the same congruence.
\end{itemize}

\subsection{Multidimensional conjugate nets and the Darboux equations}

Consider $N$-dimensional submanifold in $\RR^M$ with local parameters $\bu = (u_1, \dots , u_N)$ satisfying in each pair $(u_i, u_j)$, $i\neq j$, the conjugate net condition
\begin{equation} \label{eq:r-ij-conj}
\br_{,ij} = a_{ij} \br_{,i} + a_{ji} \br_{,j}.
\end{equation}
The functions $a_{ij}$ of the local conjugate parameters cannot be arbitrary, because for $N>2$ they should satisfy the following compatibility conditions of the above system \eqref{eq:r-ij-conj} of Laplace equations
\begin{equation} \label{eq:Darboux-a}
a_{ij,k} = a_{ij}a_{jk} + a_{ik} a_{kj} - a_{ij} a_{ik} , \qquad i,j,k \quad \text{distinct}.
\end{equation}
The nonlinear \emph{Darboux equations} \eqref{eq:Darboux-a} imply in particular that $a_{ij,k} = a_{ik,j}$, what allows to introduce potentials $h_i$, called \emph{Lam\'{e} coefficients}, such that
\begin{equation}
a_{ij} = \frac{h_{i,j}}{ h_i }, \qquad i\neq j,
\end{equation}
and correspondingly the Laplace system~\eqref{eq:r-ij-conj} takes the form
\begin{equation}
\label{eq:r-Lapl-h}
\br_{,ij} = (\log h_i)_{,j} \br_{,i} + (\log h_j)_{,i} \br_{,j}\qquad i\neq j.
\end{equation}
The remaining part of the Darboux equations in Laplace coefficients $a_{ij}$ can be written in terms of the Lam\'{e} coefficients as
\begin{equation} \label{eq:Darboux-L}
h_{i,jk}  = \frac{h_{j,k} h_{i,j}}{h_j}  + 
\frac{h_{k,j} h_{i,k}}{h_k} , \qquad i,j,k \quad \text{distinct}.
\end{equation}

Following Darboux, let us introduce the suitably scaled tangent vectors $\bX_i$, $i=1,\dots ,N$, from equations
\begin{equation} \label{eq:r-h-X}
\br_{,i} = h_i \bX_i .
\end{equation}
Then the Laplace equations \eqref{eq:r-ij-conj} take the particularly simple form
\begin{equation} \label{eq:lin-X}
\bX_{i,j} = \beta_{ij} \bX_j, \qquad i\neq j,
\end{equation}
where the \emph{rotation coefficients} $\beta_{ij}$ are defined by the linear system adjoint to~\eqref{eq:lin-X} 
\begin{equation} \label{eq:lin-h}
h_{j,i} = \beta_{ij} h_i, \qquad i\neq j.
\end{equation} 
The corresponding version of the Darboux equation reads
\begin{equation} \label{eq:Darboux-beta}
\beta_{ij,k} = \beta_{ik} \beta_{kj}, \qquad i,j,k \quad \text{distinct}.
\end{equation}

The tangent lines to $i$-th coordinate on the conjugate net $\br$ form its $i$-th tangent congruence. Its $j$-th focal net, called $ij$-Laplace transform of $\br$, is given by
\begin{equation}
\br^{(i)}_{(j)} = \br - \frac{1}{a_{ji}}\br_{,i}.
\end{equation}
The Laplace transforms satisfy the following identities
\begin{equation}
(\br^{(i)}_{(j)})^{(j)}_{(i)} =\br, \qquad (\br^{(i)}_{(j)})^{(j)}_{(k)} =\br^{(i)}_{(k)}, \qquad (\br^{(i)}_{(j)})^{(k)}_{(i)} =\br^{(k)}_{(j)}. 
\end{equation}
This means that each generic $N$-dimensional conjugate net comes together with whole system of conjugate nets enumerated by points of the $Q(A_{N-1})$ root lattice.

\subsection{The vectorial fundamental (binary Darboux) transformation \\of conjugate nets}
The contemporary theory of transformations of conjugate nets and of the Darboux equations is based on the following Lemma, given first in the discrete setting in~\cite{MDS}. The present version is its direct limit.

\begin{Lem} \label{lem:vect-fund}
	Given solution $\beta_{ij}$ of the Darboux equations~\eqref{eq:Darboux-beta}, and given solution $\bY_i$ of the linear system \eqref{eq:lin-X} taking values in the (column vector) space $\RR^K$, and given solution $\bY^*_i$ of the adjoint linear system~\eqref{eq:lin-h} taking values in the (row vector) space $\RR^L$. \\
	1. 	There exists the $K\times L$ matrix-valued potential $\bTh[\bY,\bY^*]$ defined by the following compatible system
	\begin{equation} \label{eq:Omega-i}
	\bTh[\bY,\bY^*]_{,i} = \bY_i \otimes \bY_i^*, \qquad i=1,\dots ,N.
	\end{equation}
	2. If $K=L$ and the potential $	\bTh[\bY,\bY^*]$ is invertible, then the functions 
	\begin{equation} \label{eq:vect-fund-beta}
	\hat{\beta}_{ij} = \beta_{ij} - \bY^*_j \bTh[\bY,\bY^*]^{-1} \bY_i,
	\end{equation}
	are new solutions of the Darboux equations.\\
	3. The vector-valued functions
	\begin{equation}
	\hat{\bY}_i = \bTh[\bY,\bY^*]^{-1} \bY_i, \qquad 
	\hat{\bY}^*_i = \bY^*_i \bTh[\bY,\bY^*]^{-1},
	\end{equation}
	are the corresponding new solutions of the linear and the adjoint linear problem equations~\eqref{eq:lin-X}-\eqref{eq:lin-h}. In addition, the new matrix-valued potential is of the form
	\begin{equation}
	\bTh[\hat{\bY},\hat{\bY}^*] = C - \bTh[\bY,\bY^*]^{-1},
	\end{equation}
	where $C$ is a constant operator.
\end{Lem}
\begin{Rem}
	Notice that equations \eqref{eq:r-h-X} mean that one can write $\br = \bTh[\bX,h]$.
\end{Rem}
\begin{Rem}
	In the theory of Darboux transformations of the KP hierarchy \cite{OevelSchief} the function $\bTh$ is called the squared-eigenfunction potential.
\end{Rem}
By suitable arrangement of the transformation data~\cite{TQL} we arrive at the following version of the above Lemma.
\begin{Th} Consider conjugate net $\br$ with Lam\'{e} coefficients $h_i$, normalized tangent vectors $\bX_i$ and rotation coefficients $\beta_{ij}$. Given transformation data $\bY_i$, $\bY^*_i$ which satisfy point 2. of Lemma~\ref{lem:vect-fund} then 
	\begin{equation} \label{eq:vect-fund}
	\hat\br = \br - \bTh[\bX,\bY^*] \bTh[\bY, \bY^*]^{-1} \bTh[\bY,h],
	\end{equation}
	is new conjugate net, called the fundamental transform of $\br$. The new rotation coefficients are given by equation~\eqref{eq:vect-fund-beta}, and the  Lam\'{e} coefficients and normalized tangent vectors read
	\begin{equation}
	\hat{h}_i = h_i - \bY^*_i \bTh[\bY, \bY^*]^{-1} \bTh[\bY,h],
	\qquad \hat{\bX}_i = \bX_i - \bTh[\bX,\bY^*] \bTh[\bY, \bY^*]^{-1} \bY_i.
	\end{equation}
\end{Th}

When the transformation data are scalar functions $Y_i$ and $Y^*_i$ then $\theta = \bTh[Y,h]$ is a scalar additional solution of the Laplace equation of $\br$. The $N$-dimensional vector-valued function $\bX^\prime = \bTh[\bX,Y^*]$ defines new conjugate net with Lam\'{e} coefficients $Y_i^*$ and the same normalized tangent vectors $\bX_i$. The function $\theta^\prime = \bTh[Y,Y^*]$ is a scalar solution of the Laplace equation of $\bX^\prime$
and in such a notation the scalar fundamental transformation reads
\begin{equation}
\hat{\br} = \br - \frac{\theta}{\theta^\prime}\bX^\prime.
\end{equation}
The vector $\bX^\prime$ points in the direction of the congruence of the transformation, and is called the \emph{Combescure transform} of $\br$. 
The corresponding $i$-th L\'{e}vy transform $\br^{(i)}$ of $\br$, the intersection of $i$-tangents of $\br$ and $\hat{\br}$ reads
\begin{equation}
\br^{(i)} = \br - \frac{\theta}{\theta_{,i}}\br_{,i} = \br - \frac{\theta}{Y_i}\bX_i ,
\end{equation}
while the $i$-th adjoint L\'{e}vy transform $\br_{(i)}$ of $\br$, whose $i$-th tangent congruence is the congruence of the transformation is given by
\begin{equation}
\br_{(i)} = \br - \frac{h_i}{Y^*_i} \bX^\prime.
\end{equation}
See Figure~\ref{fig:fund-uv} to compare with the two-dimensional case. 
\subsection{Superpositions of fundamental transformations, and \\multidimensional lattices of planar quadrilaterals}
The vectorial form of the fundamental transformation given in Lemma~\ref{lem:vect-fund} contains already their permutability theorem.
\begin{Th} \label{th:vect-fund}
	Assume the following splitting of the data of the vectorial fundamental transformation
	\begin{equation}
	\bY_i = \begin{pmatrix} \bY^a_i \\ \bY^b_i \end{pmatrix}, \qquad \bY^*_i = \left( \bY^*_{ai} \; \bY^*_{bi}\right) ,
	\end{equation}
	associated with the partition $\RR^K = \RR^{K_a} \oplus \RR^{K_b}$, which implies the following splitting of the potentials
	\begin{gather*}
	\bTh[\bY,h] = \begin{pmatrix} \bTh[\bY^a,h] \\ \bTh[\bY^b, h] \end{pmatrix}, \qquad 
	\bTh[\bX,\bY^*] = \left( \bTh[\bX,\bY^*_a] \; \bTh[\bX,\bY^*_b] \right),\\
	\bTh[\bY,\bY^*] = \begin{pmatrix} \bTh[\bY^a,\bY^*_a] &
	\bTh[\bY^a,\bY^*_b] \\ \bTh[\bY^b,\bY^*_a] &
	\bTh[\bY^b,\bY^*_b] \end{pmatrix}.
	\end{gather*}
	Then the vectorial fundamental transformation is equivalent to the following superposition of vectorial fundamental transformations:\\
	1. Transformation $\br \to \br^{\{ a\}} $ with the data $\bY^a_i$, $\bY^*_{ai}$ and the corresponding potentials $\bTh[\bY^a,h]$, $\bTh[\bY^a, \bY^*_{a}]$, $\bTh[\bX,\bY^*_a]$
	\begin{align}
	\br^{ \{ a \}} & = \br - \bTh[\bX,\bY^*_a] \bTh[\bY^a, \bY^*_a]^{-1} \bTh[\bY^a,h], \\
	h^{\{a\}}_i & = h_i - \bY^*_{ai} \bTh[\bY^a, \bY^*_a]^{-1} \bTh[\bY^a,h],\\
	{\bX}^{\{a\}}_i & = \bX_i - \bTh[\bX,\bY^*_a] \bTh[\bY^a, \bY^*_a]^{-1} \bY_i^a.
	\end{align}
	2. Application on the result the vectorial fundamental transformation with the transformed data
	\begin{align}
	\bY^{*\{a\}}_{bi} & = \bY^*_{bi} - \bY^*_{ai} \bTh[\bY^a, \bY^*_a]^{-1} \bTh[\bY^a,\bY^*_b],\\
	{\bY}^{b\{a\}}_i & = \bY^b_i - \bTh[\bY^b,\bY^*_a] \bTh[\bY^a, \bY^*_a]^{-1} \bY_i^a,
	\end{align}
	and potentials
	\begin{align*}
	\bTh[\bY^b,h]^{ \{ a \}} & = \bTh[\bY^b,h] - \bTh[\bY^b,\bY^*_a] \bTh[\bY^a, \bY^*_a]^{-1} \bTh[\bY^a,h] = \bTh[\bY^{b\{ a\}},h^{\{a\}}], \\
	\bTh[\bY^b,\bY^*_b]^{ \{ a \}} & = \bTh[\bY^b,\bY^*_b] - \bTh[\bY^b,\bY^*_a] \bTh[\bY^a, \bY^*_a]^{-1} \bTh[\bY^a,\bY^*_b] = \bTh[\bY^{b\{ a\}},\bY_b^{*\{a\}}],\\
	\bTh[\bX,\bY^*_b]^{ \{ a \}} & = \bTh[\bX,\bY^*_b] - \bTh[\bX,\bY^*_a] \bTh[\bY^a, \bY^*_a]^{-1} \bTh[\bY^a,\bY^*_b] = \bTh[\bX^{\{ a\}},\bY_b^{*\{a\}}],
	\end{align*}
	i.e.
	\begin{equation} \label{eq:fund-sup-r}
	\hat{\br} = \br^{\{ a,b\}} = \br^{\{a\}} - \bTh[\bX,\bY^*_b]^{ \{ a \}} \left( \bTh[\bY^b,\bY^*_b]^{ \{ a \}} \right)^{-1} \bTh[\bY^b,h]^{ \{ a \}}.
	\end{equation}
\end{Th}
\begin{Rem}
	The final result \eqref{eq:fund-sup-r} is independent of the order of making the partial transformations.
\end{Rem}
\begin{Rem}
	The above procedure fixes already the integration constants when integrating equations~\eqref{eq:Omega-i} in constructions of the potentials.
\end{Rem}
Lest us consider the simplest case of $K=2$ when the vectorial fundamental transformation is obtained as superposition of two scalar transformations
\begin{align} \label{eq:r-a}
\br^{\{a\}} & = \br - \frac{\theta^a}{\theta^a_a}\bX_a ,\qquad \bX_a = \Theta[\bX,Y^*_a], \quad \theta^a = \Theta[Y^a,h], \quad \theta^a_a = \Theta[Y^a, Y^*_a],\\
\br^{\{b\}} & = \br - \frac{\theta^b}{\theta^b_b}\bX_b ,
\qquad \bX_b = \Theta[\bX,Y^*_b], \quad \theta^b = \Theta[Y^b,h], \quad \theta^b_b = \Theta[Y^b, Y^*_b],
\end{align}
and the final result reads
\begin{equation} \label{eq:fund-sup-2}
\br^{\{a,b\}} = \br - \left( \bX_a, \bX_b \right) \begin{pmatrix} \theta^a_a & \theta ^a_b \\
\theta^b_a & \theta^b_b \end{pmatrix}^{-1} \begin{pmatrix} \theta^a \\ \theta^b \end{pmatrix}, \qquad \theta^a_b = \Theta[Y^a,Y^*_b], \quad \theta^b_a = \Theta[Y^b , Y^*_a].
\end{equation}
The point $\br^{\{a,b\}}$ belongs to the plane passing through $\br$, $\br^{\{a\}}$, $\br^{\{b\}}$ (the plane containing $\br$ and spanned by $\bX_a$ and $\bX_b$), see Figure~\ref{fig:fund-quad}, and the superposition formula \eqref{eq:fund-sup-2} can be rewritten in the form of a discrete analogue of the Laplace equation~\eqref{eq:r-ij-conj}
\begin{equation}
\br^{\{ a,b\}} - \br = \frac{\theta^{a\{b\}} \theta^a_a}{\theta^{a\{b\}}_a \theta^a} (\br^{\{ a\}} - \br) + \frac{\theta^{b\{a\}} \theta^b_b}{\theta^{b\{a\}}_b \theta^b} (\br^{\{ b\}} - \br) . 
\end{equation}
If $\br$, $\br^{\{a\}}$, $\br^{\{b\}}$ are given then the position of $\br^{\{a,b\}}$ on the plane is arbitrary due to the integration constants in definitions of $\theta^a_b$ and $\theta^b_a$. 
\begin{figure}
	\begin{center}
		\includegraphics[width=8cm]{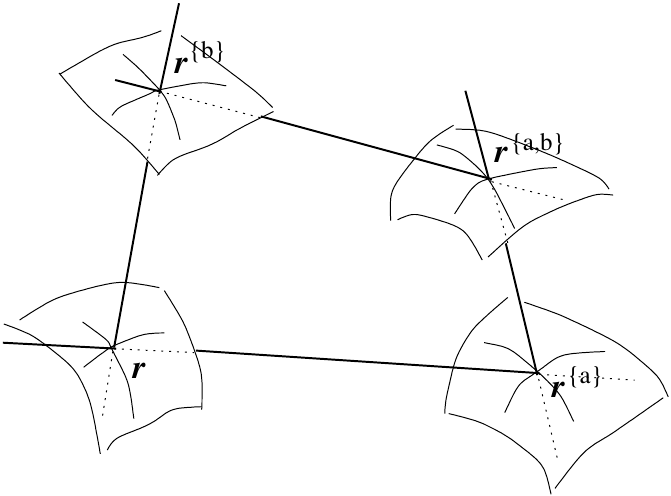}
	\end{center}
	\caption{Elementary quadrilateral of superposition of two fundamental transformations}
	\label{fig:fund-quad}
\end{figure}

The superposition of three scalar fundamental transformations
\begin{equation}
\br^{\{a,b,c\}} = \br - \left( \bX_a, \bX_b , \bX_c \right) \begin{pmatrix} \theta^a_a & \theta^a_b  & \theta^a_c\\
\theta^b_a & \theta^b_b  & \theta^b_c \\
\theta^c_a & \theta^c_b  & \theta^c_c
\end{pmatrix}^{-1} \begin{pmatrix} \theta^a \\ \theta^b \\ \theta^c \end{pmatrix}, 
\end{equation}
does not generate new integration constants. This result can be found in~\cite{Eisenhart-TS} as \emph{the extended theorem of permutability}. It implies, in particular, that  the point $\br^{\{a,b,c\}}$ is uniquely given as intersection of three planes $\langle \br^{\{a\}}, \br^{\{a,b\}}, \br^{\{a,c\}} \rangle$, $\langle \br^{\{b\}}, \br^{\{a,b\}}, \br^{\{b,c\}} \rangle$ and $\langle \br^{\{c\}}, \br^{\{a,c\}}, \br^{\{b,c\}} \rangle$; see also \cite{MQL}.

Due to equation~\eqref{eq:r-a} the vectors $\bX_a$ can be interpreted as normalized tangent vectors to the lattice direction, and $-\frac{\theta^a}{\theta^a_a}$ playing the role of the Lam\'{e} coefficients. The transformation formulas
\begin{equation} \label{eq:fund-Xa-b}
\bX_b^{\{a\}} = \bX_b - \frac{\theta^a_b}{\theta^a_a}\bX_a 
\end{equation}
give the discrete analog of the linear problem \eqref{eq:lin-X} with $-\frac{\theta^a_b}{\theta^a_a}$ playing the role of the rotation coefficients. The corresponding transformation rules provide nonlinear relations, which can be interpreted as discrete Darboux equations.
To close this Section we briefly recapitulate the basic theory of discrete conjugate nets~\cite{MQL}, which we have just obtained from the theory of transformations of conjugate nets in the spirit of works~\cite{LeBen,Levi}.

\begin{Def}
	The discrete conjugate net is a map $\br \colon \ZZ^N \to \RR^M$ of $N$-dimensional integer lattice such that for arbitrary $\bn \in \ZZ^N$ and any two directions $a \ne b$ of the lattice, the vertices $\br^{\bn}$,
	$\br^{\bn + \be_a}$, $\br^{\bn+\be_b}$, and $\br^{\bn+\be_a + \be_b}$ of elementary quadrilaterals are coplanar. 
\end{Def}
The coplanarity condition can be written in terms of the system of discrete Laplace equations
\begin{equation}  \label{eq:Laplace-dis}
\br^{\bn+ \be_a + \be_b} - \br^{\bn} = A_{ab}^{\bn} \left( \br^{\bn+\be_a } - \br^{\bn} \right) +
A_{ba}^{\bn}\left(\br^{\bn+\be_b} - \br^{\bn}\right), \qquad a\not= b. 
\end{equation}
Due to compatibility of the system \eqref{eq:Laplace-dis} 
the functions $A_{ab}^{\bn}$ can be expressed in terms of the discrete Lam\'e coefficients
\begin{equation}   \label{def:A-H}
A_{ab}^{\bn}= \frac{H_a^{\bn+\be_b}}{H_a^{\bn}} \; , \qquad a\ne b \; ,
\end{equation}
which satisfy equations
\begin{equation}
H_c^{\bn+\be_a+\be_b} - H_c^{\bn} =
\frac{H_a^{\bn+\be_b +\be_c}}{H_a^{\bn+\be_c}} \left( H_c^{\bn+\be_a} - H_c^{\bn} \right) +
\frac{H_b^{\bn+\be_a +\be_c}}{H_b^{\bn+\be_c}} \left( H_c^{\bn+\be_b} - H_c^{\bn} \right), 
\end{equation}
for distinct $a$, $b$ and $c$. If we introduce the suitably scaled tangent 
vectors $\bX_a$, $a=1,...,N$,
\begin{equation}  \label{def:HX}
\br^{\bn+\be_a} - \br^{\bn} = H_a^{\bn} \bX_a^{\bn},
\end{equation}
then equations \eqref{eq:Laplace-dis} can be rewritten as a
first order system
\begin{equation} \label{eq:lin-X-dis}
\bX_a^{\bn+\be_b} - \bX_a^{\bn} = Q_{ab}^{\bn} \bX_b^{\bn},    \qquad a\ne b \; .
\end{equation}
The proportionality factors $Q_{ab}$, called the discrete rotation coefficients, 
can be found from the linear equations
\begin{equation} \label{eq:lin-H}
H_b^{\bn+\be_a} - H_b^{\bn} = H_a^{\bn+\be_b} Q_{ab}^{\bn}, \qquad a\ne b \; ,
\end{equation}
adjoint to \eqref{eq:lin-X}.
The compatibility condition for the system~\eqref{eq:lin-X} (or its
adjoint)
gives the following form of the discrete Darboux equations
\begin{equation} \label{eq:MQL-Q}
Q_{ab}^{\bn+\be_c} - Q_{ab}^{\bn} = Q_{ac}^{\bn+\be_b}Q_{cb}^{\bn},\qquad a, b, c \quad \text{distinct}.
\end{equation}

The integrable discretization of the Darboux system was achieved first in~\cite{BoKo} within the $\bar\partial$ technique. 
The discrete analog of a conjugate net on a surface was introduced on a purely geometric basis in~\cite{Sauer2,Sauer}. Its connection with integrability theory was made first in~\cite{DCN} by connecting the Laplace sequence of two-dimensional discrete conjugate nets to Hirota's discrete generalized Toda lattice~\cite{Hirota}. Soon after that the discrete analogs of multidimensional conjugate nets were introduced in~\cite{MQL}. In particular it was shown there that the number of discrete variables in the discrete Darboux equations can be arbitrary large, and this augmentation does not restrict the solution space of the basic three-dimensional system --- this property is known nowadays as multidimensional consistency~\cite{AdlerBobenkoSuris,Nijhoff-MC} and is considered as the basic concept of the theory of discrete integrable systems~\cite{IDS}. 

The Darboux--B\"{a}cklund transformations of the discrete Darboux equations were formulated first on the algebraic level in~\cite{MDS}, and then in~\cite{TQL} the full geometric flavour of the theory was presented together with the interpretation of the transformations on the nonlocal $\bar\partial$-dressing method level, see also~\cite{LiuManas,LiuManas-SR,DMMMS,MM,MM-2}. In particular it was pointed out in~\cite{TQL}, referring to \cite{LeBen} and other similar works, that for discrete conjugate nets there is no essential difference between transformations and generation of new dimensions of the lattice. 
On the other hand, the conjugate nets are natural continuous limits of the lattices of planar quadrilaterals, what can be easily seen on the level of their Laplace equations~\eqref{eq:r-ij-conj} and \eqref{eq:Laplace-dis}. Therefore, the principle of getting the integrable discretization via B\"{a}cklund transformation approach in this particular case finds natural explanation. This shows that, from the point of view of the theory of integrable systems, the discrete ones seem to be more basic. 

The transformation theory of discrete conjugate nets can be constructed~\cite{TQL} following the geometric principles of the continuous case. Also here the notion of congruence of lines (any two neighbouring lines of the family are coplanar) turns out to be crucial. In particular, it implies natural definition of focal lattices of such congruences. The simplest congruences are given by tangent lines in a fixed direction of a discrete conjugate net.

\begin{Rem}
	The discrete Weingerten congruences, which are relevant in the theory of transformations of discrete asymptotic nets described in Section~\ref{sec:dis-ass}, do not satisfy the definition of a discrete congruences of the theory of transformations of discrete conjugate nets, as it was pointed to me by Maciej Nieszporski. In fact, as it was explained in~\cite{AD-W}, the discrete Weingarten congruences provide two-dimensional lattices of planar quadrilaterals in the Pl\"{u}cker quadric. Their theory fits therefore into the general scheme of quadratic reductions of discrete conjugate nets~\cite{q-red}. 
\end{Rem}

\section{Curvature coordinates, the Ribaucour transformation,\\ and circular lattices} \label{sec:circ}
The notion of conjugate nets is invariant with respect to projective transformations of the ambient space. By imposing an additional geometric structure one can consider the corresponding reduction of the general theory. In the Euclidean space (we work in the standard orthonormal basis where the scalar product of two vectors is given with the help of transposition $\bu \cdot \bv = \bu^t \, \bv$) one can consider orthogonal conjugate nets, which turn out to be curvature coordinates on the given submanifold. The corresponding reduction of the fundamental transformation is provided by the Ribaucour transformation. 

The classical Ribaucour transformation~\cite{Ribaucour} concerns surfaces in three dimensional Euclidean space $\EE^3$ such that lines of curvature (which are both conjugate and orthogonal) correspond and such that the normals to both surfaces in corresponding points in a point (the center of the Ribaucour sphere) equidistant to both of them. In general one can consider $N$-dimensional submanifold of $\EE^M$, $N\leq M$, parametrized by conjugate and orthogonal coordinates. 

The orthogonality constraint
\begin{equation}
\br_{,i} \cdot \br_{,j} = 0, \qquad i\neq j,
\end{equation}
implies that the function 
\begin{equation*}
\rho = \frac{1}{2}\, \br \cdot \br,
\end{equation*}
satisfies the Laplace equations of the orthogonal conjugate net $\br$, and the functions
\begin{equation*}
X^\circ_i = \br \cdot \bX_i, 
\end{equation*}
give the corresponding solution to the linear system~\eqref{eq:lin-X}, i.e. $\rho = \bTh[X^\circ , h]$.
Equivalently, the above facts imply orthogonality of the conjugate net. 

The corresponding Ribaucour reduction of the vectorial fundamental transformation, compatible with the orthogonality of a given conjugate net, can be constructed using only half of the transformation data, for analogous but different description see~\cite{LiuManas-R}. The following result can be checked directly by calculating derivatives of both sides of the formulas.
\begin{Lem}
	Given solution $\bY^*_i$ of the adjoint linear problem~\eqref{eq:lin-h} of the orthogonal conjugate net $\br$ then
	\begin{equation} \label{eq:Rib-Y-Y}
	\bY_i = \bTh[\bX,\bY^*]^t \bX_i,
	\end{equation}
	give a solution to the linear problem~\eqref{eq:lin-X} of the net. Moreover, the integration constants in construction of the corresponding potentials $\bTh[\bY,\bY^*]$, $\bTh[\bY,h]$ and $\bTh[X^\circ,\bY^*]$ can be chosen such that
	\begin{gather} \label{eq:Rib-constr}
	\bTh[\bY,\bY^*]^t + \bTh[\bY,\bY^*] = \bTh[\bX,\bY^*]^t \: \bTh[\bX,\bY^*],\\  \label{eq:Rib-constr2}
	\bTh[\bY,h] = \bTh[\bX,\bY^*]^t  \br - \bTh[X^\circ ,\bY^*]^t.
	\end{gather}
\end{Lem}
\begin{Th}
	The vectorial fundamental transformation of orthogonal conjugate net $\br$ calculated by \eqref{eq:vect-fund} with application of the above constraints \eqref{eq:Rib-constr}- \eqref{eq:Rib-constr2} is again conjugate orthogonal net.
\end{Th}
The proof is based on the observation that orthogonality of the conjugate net $\hat\br$ is equivalent to the fact that the function $\frac{1}{2} \hat{\br} \cdot \hat{\br}$ satisfies the Laplace system of the net. The statement follows then from direct verification, using the constraints \eqref{eq:Rib-constr}-\eqref{eq:Rib-constr2}, and showing  that
\begin{equation}
\frac{1}{2} \hat{\br}\cdot \hat{\br} = \frac{1}{2} \br \cdot \br  - \bTh[X^\circ,\bY^*] \bTh[\bY,\bY^*]^{-1} \bTh[\bY,h] = \hat\rho,
\end{equation}
is therefore the corresponding transform of the solution $\rho$ of the Laplace system of $\br$.

The formulation of the theorem on permutability of superpositions of vectorial Ribaucour transformations reads as in Theorem~\ref{th:vect-fund}. One has to check, what can be done by direct calculation, that on the intermediate level the reduction conditions~\eqref{eq:Rib-Y-Y}-\eqref{eq:Rib-constr2} are satisfied by the transformed data.

Geometry of integrable discrete analog of orthogonal conjugate nets follows from the observation made by
Demoulin, who showed \cite{Demoulin-R} that the vertices $\br$, $\br^{\{a\}}$, $\br^{\{b\}}$, and $\br^{\{a,b\}}$ of the Ribaucour transform are concircular, see Figure~\ref{fig:Rib-circ}. In proving that we will follow \cite{Eisenhart-TS}, where this fact was shown as an implication of the orthogonality constraint imposed on the fundamental transformation.
Given three points $\br$, $\br^{\{a\}}$ and $\br^{\{b\}}$, the coordinates of the center of the circle passing through them are of the form
\begin{equation*}
\br + \lambda \bX_a + \mu \bX_b,
\end{equation*}
where $\lambda$ and $\mu$ are determined by the condition that the lines joining the center to the mid-points of the segments $[\br , \br^{\{a\}}]$ and $[\br ,\br^{\{b\}}]$, are perpendicular to these segments. These conditions, due to the transformation formulas~\eqref{eq:r-a} and the diagonal elements of the reduction condition~\eqref{eq:Rib-constr}
\begin{equation}
2 \theta^a_a = \bX_a \cdot \bX_a ,
\end{equation} 
are reducible to
\begin{equation} \label{eq:circ-l-m}
\theta^a + (\lambda \bX_a + \mu \bX_b) \cdot \bX_a = 0, \qquad \theta^b + (\lambda \bX_a + \mu \bX_b) \cdot \bX_b = 0.
\end{equation} 
Analogously, the condition that the line joining the center to the mid-point of the segments $[\br^{\{a\}} ,\br^{\{a,b\}}]$ is perpendicular to the segment is
\begin{equation}
\left( \br^{\{ a\}} + \frac{1}{2} \left( \br^{\{ a,b \}} - \br^{\{a\}} \right) - \br -\lambda \bX_a - \mu \bX_b \right) \cdot \bX_b^{\{ a\}} = 0,
\end{equation}
what can be verified using the transformation formulas, equations~\eqref{eq:circ-l-m}, and off-diagonal elements 
of the reduction condition~\eqref{eq:Rib-constr}
\begin{equation}
\theta^a_b + \theta^b_a = \bX_a \cdot \bX_b .
\end{equation} 
\begin{Rem}
	In \cite{DMS} we proved the circularity of the quadrilaterals by showing equivalence of the constraint with equation
	\begin{equation} \label{eq:circ-X-X}
	\bX_a \cdot \bX_b^{\{ a\}} + \bX_b \cdot \bX_a^{\{ b\}} =0,
	\end{equation}
	an immediate consequence of the transformation formulas \eqref{eq:fund-Xa-b} and of the condition~\eqref{eq:Rib-constr}.
\end{Rem}
\begin{figure}
	\begin{center}
		\includegraphics[width=10cm]{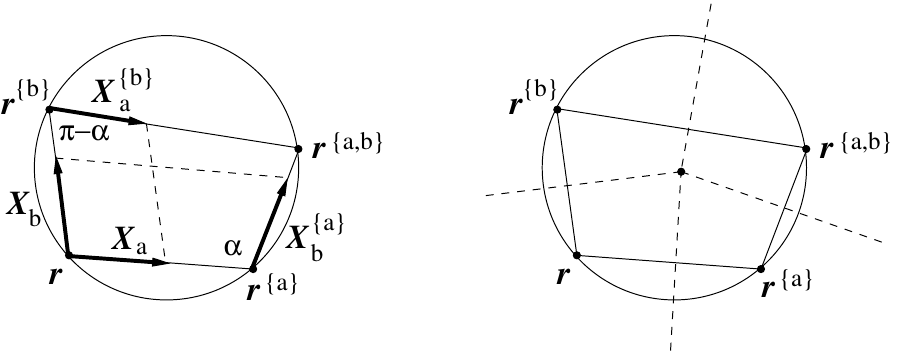}
	\end{center}
	\caption{The superposition of two scalar Ribaucour transforms of an orthogonal conjugate net. The vertices $\br$, $\br^{\{a\}}$, $\br^{\{b\}}$, and $\br^{\{a,b\}}$ of the elementary quadrilateral are concircular. For circular quadrilaterals opposite angles sum up to $\pi$. The center of the circle is the intersection point of bisectors of all the four sides of the quadrilateral}
	\label{fig:Rib-circ}
\end{figure} 

In the spirit of works~\cite{LeBen,Levi} one can conclude that the integrable discrete analogs of orthogonal conjugate nets is provided by circular lattices.

\begin{Def}
	The integrable discrete analogue of orthogonal conjugate net is a map $\br \colon \ZZ^N \to \EE^M$ of $N$-dimensional integer lattice such that for arbitrary $\bn \in \ZZ^N$ and any two directions $a \ne b$ of the lattice, the vertices $\br^{\bn}$,
	$\br^{\bn + \be_a}$, $\br^{\bn+\be_b}$, and $\br^{\bn+\be_a + \be_b}$ of elementary quadrilaterals are concircular. 
\end{Def}

Integrability of the circular reduction of lattices of planar quadrilaterals follows from the fact that the circularity constraint is preserved \cite{CDS} by evolution of the lattices expressed by the extended permutability theorem, i.e. if the three quadrilaterals with vertices
\begin{gather*}
\{ \br, \br^{\{a\}}, \br^{\{b\}}, \br^{\{a,b\}}\}, \qquad
\{ \br, \br^{\{a\}}, \br^{\{c\}}, \br^{\{a,c\}}\}, \qquad
\{ \br, \br^{\{b\}}, \br^{\{c\}}, \br^{\{b,c\}}\}, \qquad
\end{gather*}
are circular, then the three circles through the triplets
\begin{equation*}
\{ \br^{\{a\}}, \br^{\{a,b\}}, \br^{\{a,c\}} \}, \qquad
\{ \br^{\{a\}}, \br^{\{a,b\}}, \br^{\{a,c\}} \}, \qquad
\{ \br^{\{a\}}, \br^{\{a,b\}}, \br^{\{a,c\}} \}, 
\end{equation*}
intersect in the point $\br^{\{a,b,c\}}$, what is equivalent to the Miquel theorem of elementary geometry. In view of multidimensional consistency of lattices of planar quadrilaterals this implies the multidimensional consistency of circular lattices. Soon after geometric proof of  their integrability, it was confirmed by the non-local~$\bar{\partial}$ dressing technique in~\cite{DMS}, the algebro-geometric construction~\cite{AKV}, by construction of the corresponding reduction of the fundamental transformation~\cite{q-red,LiuManas-SR}, and within the free-fermion description of the KP hierarchy~\cite{DMM}.

Integrability of other basic reductions of the multidimensional lattice of planar quadrilaterals was investigated in~\cite{DS-sym}, see also~\cite{BobSur}. We remark that transformations~\cite{NiSchief} of the system of discrete Moutard equations~\eqref{eq:Mout-N} and of the corresponding Miwa's discrete BKP system \eqref{eq:BKP-nlin} can be obtained as an integrable reduction of the fundamental transformation~\cite{BQL}. For review of the theory and many other geometric aspects of integrable discrete systems see~\cite{BobSur}. Interesting applications of the theory to architectural design and computer graphics are discussed in~\cite{WallnerPottmann,WallnerPottmann2}.

\section{Conclusion and open problems}

We presented interpretation of known results in the theory of discrete asymptotic and discrete conjugate nets from the \emph{discretization by B\"{a}cklund transformations} point of view. We collected both classical formulas of XIXth century differential geometry of surfaces and their transformations, and more recent results from geometric theory of integrable discrete equations. Darboux--B\"{a}cklund transformations of difference operators are reviewed in~\cite{AD+MN}. Old ideas of differential geometry of surfaces relevant to integrability are still sources of inspiration for contemporary research, see for example \cite{Burstall,Dellinger}.

The theory of multidimensional discrete conjugate nets is based on the simple geometric principle of planarity of elementary quadrilaterals. Their integrable reductions often come from geometric understanding, in the spirit of Klein, of various reductions of projective geometry by introducing absolute objects and corresponding restrictions of the group of projective transformations. Moreover, basic analytic tools of the theory of integrable systems, like the non-local $\bar\partial$-dressing method and the algebro-geometric techniques, could be applied to construct such lattices and corresponding solutions to discrete Darboux equations in a rather pure form. 
It turns out however, surprisingly, that the principle of \emph{coplanarity of four points} can be replaced, without loss of generality from the integrability viewpoint, by the condition of \emph{collinearity of three points}~\cite{Dol-Des}. The multidimensional consistency of such lattices is equivalent to the Desargues theorem of projective geometry. 
The collinearity condition for such $N$-dimensional Desargues lattice in natural homogeneous coordinates of the projective space can be rewritten in the form
\begin{equation} \label{eq:lin-Hir}
\bpsi^{\bn+\be_i} - \bpsi^{\bn+\be_j} = u^{\bn}_{ij} \bpsi^{\bn}, \qquad 1\leq i\neq j \leq N.
\end{equation}
The compatibility condition
\begin{equation}
u^{\bn}_{ji} + u^{\bn}_{ij} = 0, \qquad u^{\bn}_{ij} + u^{\bn}_{jk} + u^{\bn}_{ki}= 0, \qquad u^{\bn}_{ij}u^{\bn+\be_j}_{ik} = u^{\bn}_{ik}u^{\bn+\be_k}_{ij}, \qquad i,j,k \quad \text{disctinct}
\end{equation}
can be simplified by introducing the single potential $\tau$ such that
\begin{equation}
u^{\bn}_{ij} = \frac{ \tau^{\bn+\be_i + \be_j} \tau^{\bn} } {\tau^{\bn+\be_i} \tau^{\bn + \be_j} } = - u^{\bn}_{ji}, \qquad i<j,
\end{equation}
which satisfies the system of Hirota equations~\cite{Hirota,Miwa}
\begin{equation} \label{eq:Hirota}
\tau^{\bn+\be_i}\tau^{\bn+\be_j + \be_k} - 
\tau^{\bn+\be_j}\tau^{\bn+\be_i + \be_k} +
\tau^{\bn+\be_k}\tau^{\bn+\be_i + \be_j} = 0, \qquad
i<j<k.
\end{equation}
As it was shown in~\cite{Dol-Des} the $2N-1$-dimensional Hirota system is equivalent to discrete Darboux equations of $N$-dimensional discrete conjugate net supplemented by $(N-1)$ dimensional lattice of its Laplace transformations.

The Hirota system~\eqref{eq:Hirota} is probably the most important discrete integrable system both from theoretical~\cite{Miwa} and practical \cite{KNS-rev} viewpoint. In particular, its multidimensional consistency gives rise to the full hierarchy of commuting symmetries of the KP equation~\cite{DKJM}. Originally it was called the discrete Toda system, which is another sign of a general rule that single integrable discrete equation can lead, via different limits, to many differential equations. In fact, the original name is related to the theory of two-dimensional lattices of planar quadrilaterals and their Laplace transforms \cite{DCN}, while the other name is related to the Desargues lattices~\cite{Dol-Des}.   Darboux--B\"{a}cklund transformations of the Hirota system were studied in~\cite{Nimmo-KP,Nimmo-rev}, where the fundamental transformation of the theory of discrete conjugate nets appears under the name of binary Darboux transformation (the L\'{e}vy transformation is called there the elementary Darboux transformation). 

Among distinguished integrable reductions of the KP hierarchy of equations there is the BKP hierarchy, which is encompassed by the single Miwa equation~\eqref{eq:BKP-nlin}. The CKP hierarchy~\cite{DJKM-CKP} leads, via the Darboux--B\"{a}cklund transformations to the corresponding reduction \cite{Kashaev-LMP,Schief-JNMP} of the Hirota system, see also \cite{DS-sym,CQL} for geometry of the corresponding reduction of the discrete conjugate nets. In parallel to the reductions of the KP hierarchy, which is based on restrictions of the Lie algebra~$\mathfrak{gl}(\infty)$ to its orthogonal and symplectic subalgebras~$\mathfrak{so}(\infty)$ and $\mathfrak{sp}(\infty)$, on the discrete level one has the corresponding root lattices and their affine Weyl group interpretations~\cite{Dol-AN,Dol-Tampa} of the Hirota system, the Miwa (discrete BKP) and the Kashaev (discrete CKP) equations. Many of known integrable systems, both discrete and continuous, can be obtained as its further reductions, see \cite{NY-AN,Sakai} for applications of the root lattices and affine Weyl groups to understand Painlev\'{e} equations. It may seem therefore that integrable discrete systems are enough to fully understand integrability.

There is however a distinguished class of integrable differential equations, called dispersionless systems~\cite{ManakovSantini,CalderbankKruglikov,Dunajski,DunajskiKrynski,KruglikovMorozov,Sergyeyev-rec},  which escapes the above interpretation. The distinguished examples of such systems are the so called generalized heavenly equations \cite{Schief-gh,DubrovFerapontov}, 
which describe self-dual Einstein spaces~\cite{Plebanski}, or the self-dual Yang--Mills equations~\cite{AblowitzChakravartyTakhtajan}, and both of them are genuine four dimensional. 



\label{lastpage}
\end{document}